%% file: pack.tex
\documentclass[10pt,twocolumn]{article}
\usepackage{latex8}
\usepackage{times}
\usepackage{cite}
\usepackage{booktabs}
\usepackage{url}
\usepackage{amsmath}
\usepackage{amssymb}
\usepackage{amsthm}
\usepackage{algorithm}
\usepackage{algorithmic}
\usepackage{graphicx}
\usepackage{subfigure}
\usepackage[english]{babel}

\include{defines}

\begin{document}
\title{Finding Good Itemsets by Packing Data}

\author{Nikolaj Tatti\\ 
HIIT, Department of Information and Computer Science \\
Helsinki University of Technology \\
ntatti@cc.hut.fi
\and
Jilles Vreeken \\
Department of Computer Science \\
Universiteit Utrecht \\
jillesv@cs.uu.nl \\
\\ 
\\ 
\\ 
} 

\maketitle
\thispagestyle{empty}
\begin{abstract}
	The problem of selecting small groups of itemsets that represent the data
	well has recently gained a lot of attention. We approach the problem by
	searching for the itemsets that compress the data efficiently. As a
	compression technique we use decision trees combined with a refined version
	of MDL.  More formally, assuming that the items are ordered, we create a
	decision tree for each item that may only depend on the previous items.
	Our approach allows us to find complex interactions between the attributes,
	not just co-occurrences of 1s.  Further, we present a link between the
	itemsets and the decision trees and use this link to export the itemsets
	from the decision trees.  In this paper we present two algorithms.  The
	first one is a simple greedy approach that builds a family of itemsets
	directly from data.  The second one, given a collection of candidate
	itemsets, selects a small subset of these itemsets.  Our experiments show
	that these approaches result in compact and high quality descriptions of
	the data.
\end{abstract}

\input{intro}
\input{prel}
\input{model}
\input{itemsets}
\input{choose}
\input{related}
\input{experiments}
\input{conclusions}

\bibliographystyle{latex8}
\bibliography{pack}

\end{document}

%% file: defines.tex
\newcommand{\ontop}[2]{\genfrac{}{}{0pt}{}{#1}{#2}}

\newcommand{\set}[1]{\left\{#1\right\}}
\newcommand{\pr}[1]{\left(#1\right)}
\newcommand{\fpr}[1]{\mathopen{}\left(#1\right)}

\newcommand{\abs}[1]{{\left|#1\right|}}

\newcommand{\enset}[2]{\left\{#1 ,\ldots , #2\right\}}
\newcommand{\enpr}[2]{\pr{#1 ,\ldots , #2}}

\newcommand{\ifam}[1]{\mathcal{#1}}
\newcommand{\iset}[2]{{{#1}\cdots{#2}}}
\newcommand{\freq}[1]{fr\fpr{#1}}

\newcommand{\define}{\leftarrow}

\newcommand{\target}[1]{t\fpr{#1}}
\newcommand{\source}[1]{src\fpr{#1}}
\newcommand{\cost}[1]{c\fpr{#1}}
\newcommand{\costmdl}[1]{c_{\textsc{MDL}}\fpr{#1}}
\newcommand{\costt}[1]{c_D\fpr{#1}}
\newcommand{\leaves}[1]{lvs\fpr{#1}}
\newcommand{\inter}[1]{intr\fpr{#1}}
\newcommand{\besttree}[1]{bt\fpr{#1}}
\newcommand{\pospath}[1]{pos\fpr{#1}}
\newcommand{\negpath}[1]{neg\fpr{#1}}
\newcommand{\sets}[1]{sets\fpr{#1}}
\newcommand{\treepath}[1]{path\fpr{#1}}

\newtheorem{theorem}{Theorem}

\newtheorem{proposition}[theorem]{Proposition}
\newtheorem{corollary}[theorem]{Corollary}
\newtheorem{example}[theorem]{Example}

%% file: intro.tex
\section{Introduction}\label{section:introduction}

One of the major topics in data mining research
is the discovery of interesting patterns in data.
From the introduction of frequent itemset mining
and association rules \cite{agrawal96}, the pattern
explosion was acknowledged: at high frequency thresholds
only common knowledge is revealed, while at low thresholds
prohibitively many patterns are returned.

Part of this problem can be solved by reducing these collections either
lossless or lossy, however even then the resulting collections are often so large
that they cannot be analyzed by hand or even machine.  Recently, it was therefore
argued \cite{han07} that while the efficiency of the search process has
received ample attention, there still exists a strong need for pattern mining
approaches that deliver compact, yet high quality, collections of patterns
(see Section~\ref{section:related} for a more detailed
discussion).
Our goal is to identify the family of itemsets that form the best description of the data.
Recent proposals to this end all consider just part of 
the data, by either only considering co-occurrences~\cite{siebes06} or 
being lossy in nature~\cite{knobbe06b,bringmannZ07,summarization}.
In this paper, we present two methods that do describe all interactions in the data.
Although different in approach, both methods return small families of
itemsets, which are selected to provide high-quality lossless descriptions of
the data in terms of local patterns.  Importantly, our parameterless methods
regard the data symmetrically.  That is, we consider not just the 1s
in the data, but also the 0s. Therefore, we are able to find patterns that
describe {\em all} interactions between items in the data, not just
co-occurrences. 

As a measure of quality for the collection of itemsets
we employ the practical variant of Kolmogorov Complexity \cite{paul}, 
the Minimum Description Length (MDL) principle \cite{peter07mdl}. 
This principle implies that we should do induction through 
compression. It states that the best model is the model that 
provides the best compression of the data: it is the model that captures best the 
regularities of the data, with as little redundancy as 
possible. 

The main idea of our approach is to use decision trees
to determine the shortest possible encoding of an attribute,
by using the values of already transmitted attributes.
For example, let us assume two binary attributes $A$ and $B$. 
Now say that for 90\% of the time when the attribute $A$ has a value of $1$, 
the attribute $B$ has a value of $0$.
If this situation occurs frequently, we
recognize this dependency, and include the item
$A$ in the tree deciding how to encode $B$. 

Using such trees allows us to find complex interactions between the items while
at the same time MDL provides us with a parameter-free framework for removing
fake interactions that are due to the noise in the data. The main outcome of 
our methods is not the decision trees, but the group of itemsets that form their paths:
these are the important patterns in the data since they capture
the dependencies between the attributes implied by the decision trees.

The two algorithms we introduce to this end are orthogonal in approach.  Our
first method builds the encoding decision trees directly from the data; it
greedily introduces splits until no split can help to compress the data
further.  Just as naturally as we can extract itemsets from these trees, we can
consider the trees that can be built from a collection of itemsets. That link
is exploited by our second method, which tries to select the best itemsets from a
larger collection.

Experimental evaluation shows that both methods return
small collections of itemsets that provide high quality 
data descriptions.
These sets allow for very short encoding of the data,
which inherently shows that the most important patterns
in the data are captured. 
As the number of itemsets are small, we can easily expose the resulting itemsets
to further analysis, either by hand or by machine.

The rest of this paper is as follows. After the covering preliminaries in 
Section~\ref{section:preliminaries}, we discuss how to use decision trees to optimally
encode the data succinct in Section~\ref{section:model}.
Next, in Section~\ref{section:itemsets}, we explain the connection between 
decision trees and itemsets. Section~\ref{section:choose} introduces
our method with which good itemsets can be selected by weighing these
through our decision tree encoding.
Related work is discussed in Section~\ref{section:related}, after which
we present the experiments on our methods in Section~\ref{section:experiments}.
We round up with discussion and conclusions in Sections~\ref{section:discussion} and \ref{section:conclusion}.

%% file: prel.tex
\section{Preliminaries and Notation}\label{section:preliminaries}
In this section we introduce preliminaries and notations used in subsequent
sections.

A \emph{binary dataset} $D$ is a collection of $\abs{D}$ \emph{transactions},
binary vectors of length $K$. The $i$th element of a random transaction is
represented by an \emph{attribute} $a_i$, a Bernoulli random variable. We
denote the collection of all the attributes by $A = \enset{a_1}{a_K}$. An
\emph{itemset} $X = \enset{x_1}{x_L} \subseteq A$ is a subset of attributes. We
will often use the dense notation $X = \iset{x_1}{x_L}$.

Given an itemset $X$ and a binary vector $v$ of length $L$, we use the notation
$p\fpr{X = v}$ to express the probability of $p\fpr{x_1 = v_1, \ldots, x_L =
v_L}$. If $v$ contains only 1s, then we will use the notation $p\fpr{X = 1}$,
if $v$ contains only 0s, then we will use the notation $p\fpr{X = 0}$.

Given a binary dataset $D$ we define $q_D$ to be an \emph{empirical distribution},
\[
q_D\fpr{A = v} = \abs{\set{t \in D \mid t = v}} / {\abs{D}}.
\]
We define the frequency of an itemset $X$ to be $\freq{X} = q_D\fpr{X = 1}$.

In the paper we use the common convention $0 \log 0 = 0$. All logarithms in the paper are of base $2$. 

In the subsequent sections we will need some knowledge of graphs. All the
graphs in the paper are directed. Given a graph $G$ we denote by $V\fpr{G}$ the
set of vertices and by $E\fpr{G}$ the edges of $G$. A directed graph is said to
be \emph{acyclic} (DAG) if there is no cycle in the graph. A directed graph 
is said to be \emph{directed spanning tree} if each node (except one special node) has exactly
one outgoing edge. The special node has no outgoing edge and is called \emph{sink}.

%% file: model.tex
\section{Packing Binary Data with Decision Trees}\label{section:model}
In this section we present our model for packing the data and a greedy
algorithm for searching good models.
\subsection{The Definition of the Model}
Our goal in this section is to define a model that is used to transmit a
binary dataset $D$ from a transmitter to a receiver.
We do this by transmitting one transaction at the time, the order of which
does not matter. Within a single transaction we transmit the items one 
at the time. 

Assume that we are transmitting an attribute $a_t$. As the attribute may have two
values, we need to have two codes to indicate its value. We define the table in
which these two codes are stored to be a {\em coding table}.
Obviously, the codes need to be optimal, that is, as short as possible.
From information theory \cite{inftheory}, we have the optimal Shannon codes 
of length $-\log(p(x))$. Here, the optimal code lengths are thus $-\log q_D\fpr{a_t = 1}$ and $-\log q_D\fpr{a_t =
0}$. We need to transmit the attribute $\abs{D}$ times. The cost of these
transmissions is
\[
-\abs{D} \sum_{v = \set{0, 1}} q_D\fpr{a_t = v}\log q_D\fpr{a_t = v}.
\]
This is the simplest case of encoding $a_t$. Note that we are 
not interested in the actual codes, but only in their lengths: they allow
us to determine the complexity of a model.

A more complex and more interesting approach to encode $a_t$ succinct 
is to have several coding tables from which the transmitter chooses one for transmission.
Choosing the coding table is done via a decision tree that branches on the values of 
other attributes in the same transaction. That is, we have a decision tree used for
encoding $a_t$ in which each leaf node is associated with a different coding table 
of $a_t$.
The leaf is selected by testing the values of other attributes within the same
transaction.

\begin{example}
Assume that we have three attributes, $a$, $b$, and $c$ and consider the trees
given in Figure~\ref{fig:extree}. In Figure~\ref{fig:extree:a} we have the
simplest tree, a simple coding table with no dependencies at all. A more
complex tree is given in Figure~\ref{fig:extree:b} where the transmitter
chooses from two coding table for $a$ based on the value of $c$.
Similarly in, Figure~\ref{fig:extree:d} we have three different coding tables
for $c$. The choice of the coding table in this case is based on the values of
$a$ and $b$.

\begin{figure}[htbp!]
\center
\subfigure[$T_1$, Trivial tree encoding $a$]{\label{fig:extree:a}\makebox[4cm]{\includegraphics[scale=0.4]{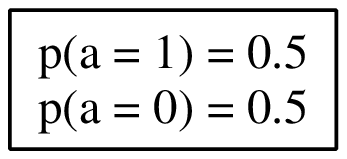}}}
\subfigure[$T_2$, Alternative tree for $a$]{\label{fig:extree:b}\makebox[4cm]{\includegraphics[scale=0.4]{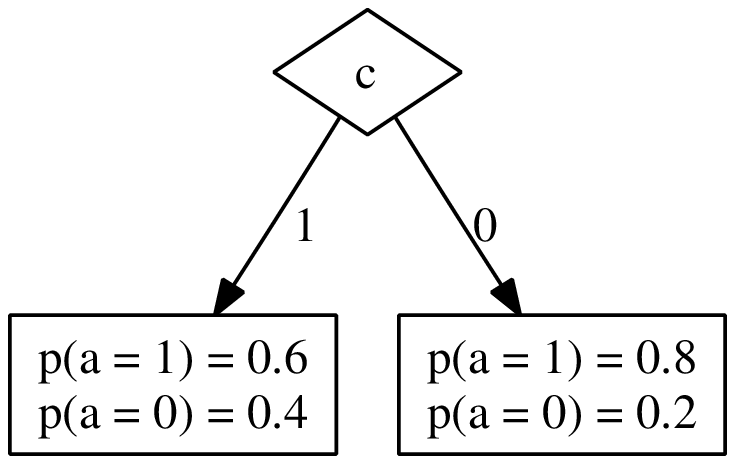}}}\\
\subfigure[$T_3$, Tree for $b$]{\label{fig:extree:c}\includegraphics[scale=0.4]{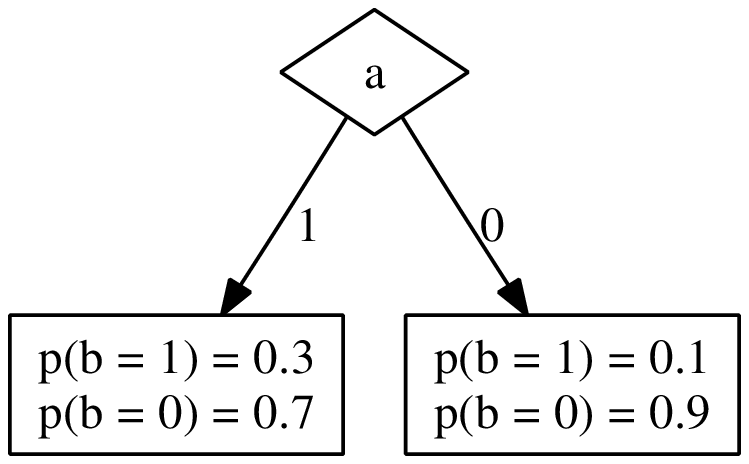}}
\subfigure[$T_4$, Tree for $c$]{\label{fig:extree:d}\includegraphics[scale=0.4]{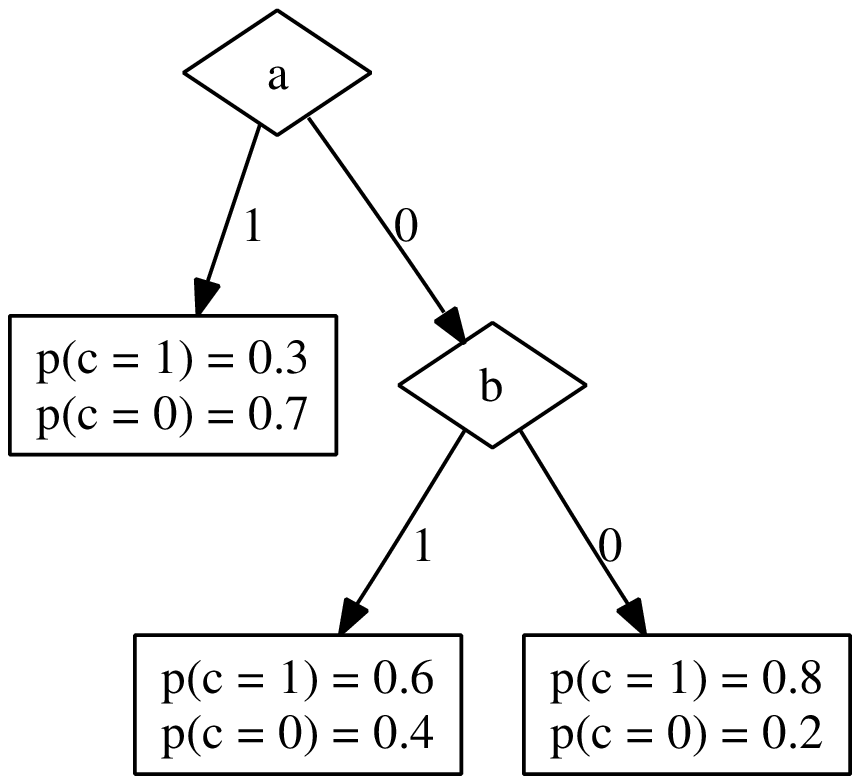}}
\label{fig:extree}
\caption{Toy decision trees.}
\end{figure}
\end{example}

Let us introduce some notation. Let $T$ be a tree encoding $a_t$. We use
the notation $\target{T} = a_t$. We set $\source{T}$ to be the set
of all items used in $T$ for choosing the coding table.

\begin{example}
For the tree $T_3$ in Figure~\ref{fig:extree:c} we have $\target{T_3} = b$
and $\source{T_3} = \set{a}$ and for $T_4$ in Figure~\ref{fig:extree:d} we
have $\target{T_4}$ and $\source{T_4} = \set{a, b}$.
\end{example}

To define the cost of transmitting $a_t$ we first define $\leaves{T}$ to be the
set of all leaves in $T$. Let $L \in \leaves{T}$ be a leaf and $q_D\fpr{L}$ be the
probability of $L$ being chosen. Further, $q_D\fpr{a_t = v \mid L}$ is the probability of
$a_t = v$ given that $L$ is chosen. We now know that the optimal cost, denoted by $\costt{T}$, is
\[
-\abs{D} \sum_{L \in \leaves{T}} \sum_{v = \set{0, 1}} q_D\fpr{a_t = v, L}\log q_D\fpr{a_t = v \mid L}.
\]

\begin{example}
The number of bits needed by $T_1$ in Figure~\ref{fig:extree:a} to transmit
$a$ in a random transaction is
\[
-0.5\log 0.5 - 0.5 \log 0.5 = 1.
\]
Similarly, if we assume that $q_D(a = 1) = q_D(a = 0) = 0.5$, the number of bits
needed by $T_3$ to transmit $c$ in a random transaction is
\[
\begin{split}
&0.5\pr{-0.3\log 0.3 - 0.7 \log 0.7} + \\
&\quad 0.5\pr{-0.1\log 0.1 - 0.9 \log 0.9} = 0.62.
\end{split}
\]
\end{example}

In order for the receiver to decode the attribute $a_t$ he must know
what coding table was used. Thus, he must be able to use the same decision
tree that the transmitter used for encoding $a_t$. To ensure this, the transmitter
must know $\source{T}$ when decoding $a_t$. So, the attributes must have
an order in which they are sent \emph{and} the decision trees may only use
the attributes that have already been transmitted.

The aforementioned requirement is easily characterized by the following
construction. Let $G$ be a directed graph with $K$ nodes, each node
corresponding to an attribute. The graph $G$ contains all the edges of form
$\pr{a_t, a_s}$ where $a_s \in \source{T}$, where $T$ is the tree encoding
$a_t$. We call $G$ the \emph{dependency graph}. It is easy to see that there
exists an order of the attributes if and only if $G$ is an acyclic
graph (DAG). If $G$ constructed from a set of trees $\mathcal{T} =
\enset{T_1}{T_K}$ is indeed DAG we call the set $\mathcal{T}$ a \emph{decision
tree model}.

\begin{example}
Consider a graph given in Figure~\ref{fig:exdag:a} constructed from the trees
$T_2$, $T_3$, and $T_4$ (Figure~\ref{fig:extree}). We cannot use this
combination of trees for encoding since there is a cycle in the graph. On the other
hand if we use trees $T_1$, $T_3$, and $T_4$, then the resulting graph (given
in Figure~\ref{fig:exdag:b}) is acyclic and thus these trees can be used for
the transmission.
\begin{figure}[htbp!]
\center
\subfigure[Dependency graph with cycles]{\label{fig:exdag:a}\makebox[4.3cm]{\includegraphics[scale=0.6]{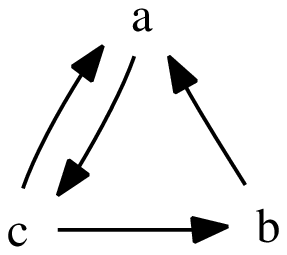}}}%
\subfigure[Dependency acyclic graph]{\label{fig:exdag:b}\makebox[4.3cm]{\includegraphics[scale=0.6]{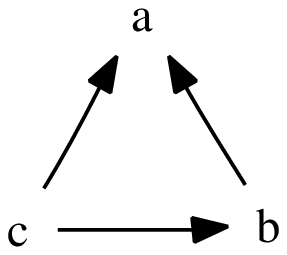}}}
\label{fig:exdag}
\caption{Dependency graphs constructed from the trees given in Figure~\ref{fig:extree}.}
\end{figure}
\end{example}

\subsection{Encoding the Data}

In order for the receiver to be able to decode the attributes, he must
know both the coding tables and the trees. Hence, we need to transmit
both of these. First, we cover how the coding tables, the leafs of the 
decision trees, are transmitted.

To transmit the coding tables we use the concept of Refined MDL~\cite{peter07mdl}.
Refined MDL is an improved version of the more traditional two-part MDL (sometimes
referred to as the crude MDL). The basic idea of the refined variant is that 
instead of transmitting the coding tables, the transmitter and the receiver 
use so called universal codes. Universal codes are the cornerstone of Refined MDL.
As these are codes can be derived without any further shared information,
this allows for a good weighing of the actual complexity of the data and model,
with virtually no overhead.
While the practicality of applying such codes depends on the type of the model,
our decision trees are particularly well-suited.

These universal codes provide a cost called the \emph{complexity} of the model.
This cost can be calculated as follows: let $L$ be a leaf in the decision tree (i.e. coding table),
and $M$ be the number of transactions for which $L$ is used. 
Then the complexity of this leaf, denoted by $\costmdl{L}$, is
\[
\costmdl{L} = \log \sum_{k = 0}^M \pr{\ontop{M}{k}}\pr{\frac{k}{M}}^k\pr{\frac{M-k}{M}}^{M - k} .
\]
In general, there is no known closed formula for the complexity of the model.
Hence estimates are usually employed~\cite{rissanen96fisher}.  However, for our
tree models we can apply an existing linear-time algorithm that solves the complexity
for multinomial models~\cite{kontkanen07linear}. 
We should also point out that the Refined MDL
is asymptotically equivalent to Bayes Information Criteria (BIC) if the number
of transactions goes to infinity and the number of free parameters stays fixed.
However, for moderate numbers of transactions there may be significant
differences~\cite{peter07mdl}.

Now that the coding tables can be transmitted,
we need to know how to transmit the actual tree $T$. To encode the tree we simply
transmit the nodes of the tree in a sequence. We use one bit to indicate
whether the node is a leaf, or an intermediate node $N \in \inter{T}$. 
For an intermediate node we additionally use $\log K$ bits, where $K$ is the
number of attributes in $D$, to indicate the item that is used for the split.

The combined cost of a tree $T$, denoted by $\cost{T}$, is
\[
\begin{split}
\cost{T} = & \sum_{N \in \inter{T}}{\bigl( 1 + \log K\bigr)} \\
           & + \costt{T} + \sum_{L \in \leaves{T}} {\bigl(1 + \costmdl{L} \bigr)},
\end{split}
\]
that is, the cost $\cost{T}$ is the number of bits needed to transmit the
tree \emph{and} the attribute $a_t$ in each transaction of $D$.

\begin{example}
Assume that we have a dataset with $100$ transactions and $3$ items.
Assume also that $q_D\fpr{a = 0} = q_D\fpr{a = 1} = 0.5$.
We know that the complexity of the leaves in this case is $\costmdl{L} = 3.25$.
The cost of the tree $T_3$ (Figure~\ref{fig:extree:c} is
\[
\begin{split}
\cost{T_3} = & 1 + \log 3 \\
& + 1 + 3.25 + 50\pr{-0.3\log 0.3 - 0.7 \log 0.7} \\
& + 1 + 3.25 + 50\pr{-0.1\log 0.1 - 0.9 \log 0.9} \\
= & 69.8.
\end{split}
\]
\end{example}

Given a decision tree model $\mathcal{T} = \enset{T_1}{T_K}$ we define the cost
$\cost{\mathcal{T}} = \sum_i\cost{T_i}$. The cost $\cost{\mathcal{T}}$ is the
number of bits needed to transmit the trees, one for each attribute, and the
complete dataset $D$.

We should point out that for data with many items, the term $\log K$ grows
and hence the threshold increases for selecting an attribute into any decision tree.
This is an interesting behavior, as due to the finite number of
transactions, for datasets with many items there is an increased probability that
two items will correlate, even though they are independent according to the
generative distribution.

\subsection{Greedy Algorithm}
Our goal is to find the decision tree model with the lowest complexity cost. 
However, since
many problems related to the decision trees are \textbf{NP}-complete~\cite{murthy96thesis}
we will resort to a greedy heuristic to approximate the decision tree model
$\mathcal{T}$ with the lowest $\cost{\mathcal{T}}$. 
It is based on the ID3 algorithm.

To fully introduce the algorithm we need some notation: By
$\textsc{TrivialTree}\fpr{a_t}$ we mean the simplest tree packing $a_t$ without
any other attributes (see Figure~\ref{fig:extree:a}). Given a tree
$T$, a leaf $L \in \leaves{T}$, and an item $c$ not occurring in the path from
$L$ to the root of $T$, we define $\textsc{SplitTree}\fpr{T, L, c}$ to be a new
tree where $L$ is replaced by a non-leaf node testing the value of $c$ and
having two leaves as the branches.

The algorithm \textsc{GreedyPack} starts with a tree model consisting only of
trivial trees.  The algorithm finds the tree which saves the most bits by
splitting. To ensure that the decision tree model is valid,
\textsc{GreedyPack} builds a dependency graph $G$ describing the dependencies
of the trees and makes sure that $G$ is acyclic. The algorithm terminates when
no further split can be made that saves any bits.

\begin{algorithm}[ht!]
\caption{\textsc{GreedyPack} algorithm constructs a decision tree model
$\mathcal{T} = \enset{T_1}{T_K}$ from a binary data $D$.}
\begin{algorithmic}[1]
\STATE $V \define \enset{v_1}{v_K}$, $E \define \emptyset$.
\STATE $G \define \pr{V, E}$.
\STATE $T_i \define \textsc{TrivialTree}\fpr{a_i}, \text{ for } i = 1, \ldots, K$.
\WHILE{there are changes}
\FOR{$i = 1, \ldots, K$}
\STATE $O_i \define T_i$.
\FOR{$L \in \leaves{T_i}$, $j = 1, \ldots, K$}
\IF{$E \cup \pr{v_i, v_j}$ is acyclic \textbf{and} $a_j \notin \treepath{L}$}
\STATE $U \define \textsc{SplitTree}\fpr{T_i, L, a_j}$.
\IF{$\cost{U} < \cost{O_i}$}
\STATE $O_i \define U$,  $s_i \define j$.
\ENDIF
\ENDIF
\ENDFOR
\ENDFOR
\STATE $k \define \arg \min_i \set{\cost{O_i} - \cost{T_i}}$.
\IF{$\cost{O_k} < \cost{T_k}$}
\STATE $T_k \define O_k$.
\STATE $E \define E \cup \pr{v_k, v_{s_k}}$.
\ENDIF
\ENDWHILE
\RETURN $\enset{T_1}{T_K}$.
\end{algorithmic}
\end{algorithm}

%% file: itemsets.tex
\section{Itemsets and Decision Trees}\label{section:itemsets}
So far we have discussed how to transmit binary data by using decision trees. In
this section we present how to select the itemsets representing the
dependencies implied by the decision trees. We will use this link in Section~\ref{section:choose}.
A similar link between itemsets and decision trees is explored in~\cite{nijssen07decision} although
our setup and goals are different.

Given a leaf $L$, the dependency of the item $a_t$ is captured in the coding
table of $L$. Hence we are interested in finding itemsets that carry the same
information. That is, itemsets from which we can compute the coding table.
To derive the codes for the leaf $L$ it is sufficient to compute the probability
\begin{equation}
q_D\fpr{a_t = 1 \mid L} = q_D\fpr{a_t = 1, L} / q_D\fpr{L}.
\label{eq:coding}
\end{equation}

Our goal is to express the probabilities on the right side of the equation using
itemsets. In order to do that let $P$ be the path from $L$ to its root. Let
$\pospath{L}$ be the items along the path $P$ which are tested
positive. Similarly, let $\negpath{L}$ be the attributes which are tested
negative. Using the inclusion-exclusion principle we see that
\begin{equation}
\begin{split}
q_D\fpr{L} & = q_D\fpr{\pospath{L} = 1, \negpath{L} = 0} \\
& = \sum_{V \subseteq \negpath{L}} (-1)^{\abs{V}} \freq{\pospath{L} \cup V}.
\end{split}
\label{eq:itemsets}
\end{equation}
We compute $q_D\fpr{a_t = 1, L}$ in a similar fashion. Let us define 
$\sets{L}$ for a given leaf $L$ to be
\[
\begin{split}
\sets{L} = & \set{V \cup \pospath{L} \mid V \subseteq \negpath{L}} \\
\cup & \set{V \cup \pospath{L} \cup \set{a_t} \mid V \subseteq \negpath{L}}.
\end{split}
\]
Combining Eqs.~\ref{eq:coding}--\ref{eq:itemsets} we see that the collection
$\sets{L}$ satisfies our goal.

\begin{proposition}
The coding table associated with the leaf $L$ can be computed from the
frequencies of $\sets{L}$.
\end{proposition}

\begin{example}
Let $L_1$, $L_2$, and $L_3$ be the leaves (from left to right) of $T_4$ in
Figure~\ref{fig:extree:d}. Then the corresponding families of itemsets are
$\sets{L_1} = \set{a, ac}$,
$\sets{L_2} = \set{b, ab, bc, abc}$, and
$\sets{L_3} = \set{\emptyset, a, b, ab, c, ac, bc, abc}$.
\end{example}

We can easily see that
the family $\sets{L}$ is essentially the smallest family of itemsets from which
the coding table can be derived uniquely.
\begin{proposition}
Let $\ifam{G} \neq \sets{L}$ be a family of itemsets. Then there are two data
sets, say $D_1$ and $D_2$, for which $q_{D_1}\fpr{a_t = 1 \mid L} \neq
q_{D_2}\fpr{a_t = 1 \mid L}$ but $\freq{\ifam{G}; D_1} = \freq{\ifam{G}; D_2}$.
\end{proposition}

Given a tree $T$ we define $\sets{T}$ to be
$\sets{T} = \bigcup_{L \in \leaves{T}} \sets{L}$.
We also define $\sets{\mathcal{T}} = \bigcup_i \sets{T_i}$ where $\mathcal{T} =
\enset{T_1}{T_K}$ is a decision tree model.

%% file: choose.tex
\section{Choosing Good Itemsets}\label{section:choose}

The connection between itemsets and decision trees made
in the previous section allows us to consider an orthogonal
approach to identify good itemsets.
Informally, our goal is to construct decision trees
from a family of itemsets $\ifam{F}$, selecting the subset 
from $\ifam{F}$ that provides the best compression of the data.
More formally, our new approach is as follows:
given a downward closed family of itemsets
$\ifam{F}$, we build a decision tree model $\mathcal{T} =
\enset{T_1}{T_K}$ providing a good compression of the data, with 
$\sets{\mathcal{T}} \subseteq \ifam{F}$.

Before we can describe our main algorithm, we need to introduce
some further notation. 
Firstly, given two trees $T_p$ and $T_n$ not using attribute $c$, we define
$\textsc{JoinTree}\fpr{c, T_p, T_n}$ to be the join tree with $c$ as the root
node, $T_p$ as the positive branch of $c$, and $T_n$ as the negative branch of
$c$. 
Secondly, to define our search algorithm we need to find the best tree
\[
\begin{split}
\besttree{a_t; S, \ifam{F}} = &\arg\min_T \left\{\cost{T} \mid \target{T} = a_t, \right. \\
& \quad \left.\source{T} \subseteq S, \sets{T} \subseteq \ifam{F}\right\},
\end{split}
\]
that is, $\besttree{a_t; S, \ifam{F}}$, returns the best tree for $a_t$
for which the related sets are in $\ifam{F}$ and only splits on attributes in $S$.

To compute the optimal tree $\besttree{a_t; S, \ifam{F}}$, we
use the exhaustive method (presented originally in~\cite{nijssen07decision}) given in
Algorithm~\ref{alg:generate}. The algorithm is straightforward:
it tests each valid item as the root 
and recurses itself on both branches.

\begin{algorithm}[ht!]
\caption{\textsc{Generate} algorithm for calculating $\besttree{a_t; S,
\ifam{F}}$, that is, the best tree $T$ for $a_t$ using only $S$ as source and
having $\sets{T} \subseteq \ifam{F}$.}
\label{alg:generate}
\begin{algorithmic}[1]
\STATE $B \define S \cap \pr{\bigcup \ifam{F}}$.
\STATE $\mathcal{C} \define \textsc{TrivialTree}\fpr{a_t}$.
\FOR{$b \in B$}
\STATE $\ifam{G} \define \set{X - b \mid b \in X \in \ifam{F}}$.
\STATE $(D_p, D_n) \define \textsc{Split}\fpr{D, b}$.
\STATE $T_p \define \textsc{Generate}\fpr{a_t, \ifam{G}, S, D_p}$.
\STATE $T_n \define \textsc{Generate}\fpr{a_t, \ifam{G}, S, D_n}$.
\STATE $\mathcal{C} \define \mathcal{C} \cup \textsc{JoinTree}\fpr{b, T_p, T_n}$.
\ENDFOR
\RETURN $\arg \min_T \set{\cost{T} \mid T \in \mathcal{C}}$.
\end{algorithmic}
\end{algorithm}

We can now describe the actual algorithm for constructing decision tree models with a low cost.
Our method automatically discovers the order in which the attributes
can be transmitted most succinct. For this, it needs to find sets of attributes
$S_i$ for each attribute $a_i$ such that these should be encoded before $a_i$. 
The collection $\mathcal{S} = \enset{S_1}{S_K}$ should define an acyclic
graph and the actual trees are $\besttree{a_i; S_i, \ifam{F}}$. 
We use $\cost{\mathcal{S}}$ as a shorthand for the total complexity 
$\sum_i \cost{\besttree{a_i; S_i, \ifam{F}}}$ of the best model built from 
$\mathcal{S}$.

We construct the set $\mathcal{S}$ iteratively. At the beginning of the
algorithm we have $S_i = \emptyset$ and we increase the sets $S_i$ one
attribute at a time. We allow ourselves to mark the attributes. The idea is
that once the attribute $a_i$ is marked, then we are not allowed to augment
$S_i$ any longer.  At the beginning none of the nodes are marked.  

To describe a single step in the algorithm we consider a graph $H =
\enpr{v_0}{v_K}$, where $v_1, \ldots, v_K$ represent the attributes and $v_0$
is a special auxiliary node. 
We start by adding edges $\pr{v_i, v_0}$ having the weight 
$\cost{\besttree{a_i; S_i, \ifam{F}}}$, thus the cost of the best
tree possible from $\ifam{F}$ using only the attributes in $S_i$.
Then, for each unmarked node $v_i$ we find out what other
extra attribute will help most to encode it succinct. 
To do this, we add the edge $\pr{v_i, v_j}$ for each $v_j$ with the weight $\cost{\besttree{a_i; S_i \cup \set{a_j}, \ifam{F}}}$. 
Now, let $U$ be the minimum directed spanning tree of $H$ having $v_0$ as the sink. 
Consider an unmarked node $v_i$ such that $\pr{v_i, v_0} \in E\fpr{U}$. 
That node is now the best choice to be fixed, as it helps to encode
the data best. We therefore mark attribute $a_i$ and add $a_i$ to
each $S_j$ for each ancestor $v_j$ of $v_i$ in $U$. 
This process is repeated until all attributes are marked. 
The details of the algorithm are given in Algorithm~\ref{alg:choose}.

\begin{algorithm}[htb!]
\caption{The algorithm \textsc{SetPack} constructs a decision tree model
$\mathcal{T}$ given a family of itemsets $\ifam{F}$ such that
$\sets{\mathcal{T}} \subseteq \ifam{F}$. Returns a DAG, a family $S =
\enpr{S_1}{S_K}$ of sets of attributes. The trees are
$T_i = \besttree{a_i, S_i, \ifam{F}}$.}
\label{alg:choose}
\begin{algorithmic}[1]
\STATE $\mathcal{S} = \enpr{S_1}{S_K} \define \enpr{\emptyset}{\emptyset}$.
\STATE $r = \enpr{r_1}{r_K} \define \enpr{\FALSE}{\FALSE}$.
\STATE $V \define \enset{v_0}{v_K}$.
\WHILE{there exists $r_i = \FALSE$}
\STATE $E \define \emptyset$.
\FOR{$i = 1, \ldots, K$}
\STATE $E \define E \cup \pr{v_i, v_0}$.
\STATE $w\fpr{v_i, v_0} \define \cost{\besttree{a_i; S_i, \ifam{F}}}$.
\IF{$r_i = \FALSE$}
\FOR{$j = 1, \ldots, K$}
\STATE $T \define \besttree{a_i; S_i \cup \set{a_j}, \ifam{F}}$.
\IF{$\cost{T} \leq w\fpr{v_i, v_0}$}
\STATE $E \define E \cup \pr{v_i, v_j}$, $w\fpr{v_i, v_j} \define \cost{T}$.
\ENDIF
\ENDFOR
\ENDIF
\ENDFOR
\STATE $U \define dmst\fpr{V, E}$ \COMMENT{Directed Min. Spanning Tree.}
\FOR{$\pr{v_i, v_0} \in E\fpr{U}$ \textbf{and} $r_i = \FALSE$}
\STATE $r_i \define \TRUE$.
\FOR{$v_j$ is a parent of $v_i$ in $U$}
\STATE $S_j \define S_j + a_i$.
\ENDFOR
\ENDFOR
\ENDWHILE
\RETURN $S$.
\end{algorithmic}
\end{algorithm}

The marking of the attributes guarantees that there can be no cycles in
$\mathcal{S}$. In fact, the marking order also tells us a valid order
for transmitting the attributes. Further, as at least one attribute is
marked at each step, this guarantees that the algorithm 
terminates in $K$ steps.

Let $\mathcal{S}$ be the collection of sources. 
The following proposition tells
us that the augmentation performed by \textsc{SetPack} 
does not compromise the optimality of collections next to $\mathcal{S}$.

\begin{proposition}
Assume the collection of sources $\mathcal{S} = \enset{S_1}{S_K}$. Let
$\mathcal{O} = \enset{O_1}{O_K}$ be the collection of sources 
such that $S_i \subseteq O_i$ and
$\abs{O_i} \leq \abs{S_i} + 1$. Let $\mathcal{S}'$ be the collection that
Algorithm~\ref{alg:choose} produces from $\mathcal{S}$ in a single step. Then
there is a collection $\mathcal{S}^*$ such that $S'_i \subseteq S^*_i$ and that
$\cost{\mathcal{S}^*} \leq \cost{\mathcal{O}}$.
\end{proposition}

\begin{proof}
Let $G$ be the graph constructed by Algorithm~\ref{alg:choose} for the
collection $S$. Construct the following graph $W$: For each $O_i$ such that
$O_i = S_i$ add the edge $\pr{v_i, v_0}$. For each $O_i \neq S_i$ add the edge
$\pr{v_i, v_j}$, where $\set{a_j} = O_i - S_i$. But $W$ is a directed spanning
tree of $G$. Let $U$ be the directed minimum spanning tree returned by the
algorithm. Let $S^*_i = S'_i$ if $\pr{v_i, v_0} \in E\fpr{U}$ and $S^*_i = S'_i
\cup \set{a_j}$ if $\pr{v_i, v_j} \in E\fpr{U}$. Note that $\mathcal{S}^*$
defines a valid model and because $U$ is optimal we must have
$\cost{\mathcal{S}^*} \leq \cost{\mathcal{O}}$.
\end{proof}

\begin{corollary}
Assume that $\ifam{F}$ is a family of itemsets having 2 items, at maximum.
The algorithm \textsc{SetPack} returns the optimal tree model.
\end{corollary}

Let us consider the complexity of the algorithms. The algorithm
\textsc{SetPack} runs in a polynomial time. By using dynamic programming
we can show that \textsc{Generate} runs in $O(\abs{F}^2)$ time.
We also tested a faster variant of the algorithm in which the exhaustive
search in \textsc{Generate} is replaced by the greedy approach similar
to the \textsc{ID3} algorithm. We call this variant \textsc{SetPackGreedy}.

%% file: related.tex
\section{Related Work}\label{section:related} Finding interesting itemsets is a
major research theme in data mining. To this end, many measures have been suggested over
time.  A classic measure for ranking itemsets is frequency, for which
there exist efficient search algorithms~\cite{agrawal96,han00mining}.  Other
measures involve comparing how much an itemset deviates from the
independence assumption~\cite{brin97beyond,aggarwal98new,dumouchel01empirical,brin97itemsets}.
In yet other approaches more flexible models are used, such as, Bayes
networks~\cite{jaroszewicz04bayes,jaroszewicz05bayes}, Maximum Entropy
estimates~\cite{meo00dependence,tatti08rank}.  Related are also low-entropy
sets: itemsets for which the entropy of the data is low~\cite{heikinheimo07entropy}.

Many of these approaches suffer from the fact that they require a user-defined
threshold and further that at low thresholds extremely many
itemsets are returned, many of which convey the same information. 
To address the latter
problem we can use closed~\cite{pasquier99discovering} or
non-derivable~\cite{calders02mining} itemsets that provide a concise
representation of the original itemsets. However, these methods
deteriorate even under small amounts of noise.

Alternative to these approaches of describing the pattern set, 
there are methods that instead pick groups of itemsets that describe the data well. 
As such, we are not the first to embrace the compression approach to data
mining~\cite{faloutsos07}.
Recently, Siebes et al.~\cite{siebes06} introduced the MDL-based {\sc Krimp}
algorithm to battle the frequent itemset explosion at low support thresholds.
It returns small subsets of itemsets that together capture the
distribution of the data well.  These {\em code tables} have been
successfully applied in classification~\cite{leeuwen06}, 
measuring the dissimilarity of data~\cite{vreeken07a}, and
data generation~\cite{vreeken07b}.
While these applications shows the practicality of the approach, {\sc Krimp}
can only describe the patterns between the items that are present in the
dataset.  On the other hand, we consider the $0$s and the
$1$s in the data symmetrically and hence we are able to provide more
detailed descriptions of the data; including patterns between the presence and
absence of items.

More different from our methods are the lossy data description approaches.  These
strive to describe just part of the data, and as such may overlook important
interactions.  Summarization~\cite{summarization} is a compression approach
that identifies a group of itemsets such that each transaction is summarized by
one set with as little loss of information as possible.  Yet different 
are \emph{pattern teams}~\cite{knobbe06b}, which are groups of
most-informative length-$k$ itemsets~\cite{knobbe06a}, selected through an
external interestingness measure.  As this approach is computationally
intensive, the number of team members is typically $<10$.
Bringmann et al.~\cite{bringmannZ07} proposed a similar selection method
that can consider larger pattern sets. However, it also requires the user to
choose a quality measure to which the pattern set has to be optimized, unlike
our parameter-free and lossless method.

Alternatively we can view the approach in this paper as building a global model
for data and then selecting the itemsets that describe the model. This
approach then allows us to use MDL as a model selection technique. In a related
work~\cite{tatti08partition} the authors build decomposable models in order to
select a small family of itemsets that model the data well.

The decision trees returned by our methods, and particularly the DAG that they
form, have a passing resemblance to Bayes networks~\cite{Bayes}. However, as
both the model construction and complexity weighing differ strongly, so do the
outcomes. To be more precise, in our case the distributions $p\fpr{x,
par\fpr{x}}$ are modeled and weighted via decision trees whereas in the Bayes
network setup any distribution is weighted equally. Furthermore, we use the
correspondence between the itemsets and the decision trees to output local
patterns, as opposed to Bayes networks which are traditionally used as global
models.

%% file: experiments.tex
\section{Experiments}\label{section:experiments}

This section contains the results of the empirical evaluation of
our methods using toy and real datasets.

\subsection{Datasets}

For the experimental validation of the two packing
strategies we use a group of datasets with strongly
differing statistics.
From the LUCS/KDD repository~\cite{coenen03} we took a number of 
often used databases to allow for comparison to other methods.
To test our methods on real data we used the Mammals presence database
and the Helsinki CS-courses dataset. 
The latter contains the enrollment records of 
students taking courses at the Department of Computer Science
of the University of Helsinki. 
The {\em mammals} dataset consists of the absence/presence of European
mammals~\cite{mitchell-jones99} in geographical areas of 50x50 kilometers.\footnote{The full version of the dataset is available for research purposes upon request, \url{http://www.european-mammals.org}.}
The details of these datasets are provided in Table~\ref{table:datasets}.

\begin{table}[htbp!]
\begin{center}
\caption{Statistics of the datasets used in the experiments.}\label{table:datasets}
\begin{tabular}[b]{lrrr}
\toprule
{\em Dataset} & {$|D|$} & {$K$} & {\% of 1's} \\
\midrule
{\em anneal} & 898 & 71 & 20.1 \\
{\em breast} & 699 & 16 & 62.4\\
{\em courses} & 3506 & 98 & 4.6\\
{\em mammals} & 2183 & 40 & 46.9 \\
{\em mushroom} & 8124 & 119 & 19.3 \\
{\em nursery} & 12960 & 32 & 28.1 \\
{\em pageblocks} & 5473 & 44 & 25.0 \\
{\em tic--tac--toe} & 958 & 29 & 34.5 \\
\bottomrule
\end{tabular}
\end{center}
\end{table}

\subsection{Experiments with Toy Datasets}

To evaluate whether our method correctly identifies
(in)dependencies, we start our experimentation using two
artificial datasets of 2000 transactions and 10 items. 
For both databases, the data is generated per transaction,
and the presence of the first item is based on a fair
coin toss. For the first database, the other items are
similarly generated. However, for the second database, 
the presence of an item is 90\% dependent on the previous item. 
As such, both datasets have item densities of about 50\%.

\begin{table*}[t!]
\begin{center}
\caption{Compression, number of trees and numbers of extracted itemsets for the greedy algorithm.}\label{table:greedy}
\begin{tabular}[b]{l rrrrr@{\hspace{3mm}}c@{\hspace{3mm}}rrrr}
\toprule
& \multicolumn{5}{l}{\sc GreedyPack} && \multicolumn{4}{l}{\sc Krimp} \\
\cmidrule{2-6} \cmidrule{8-11}
{\em Dataset} &
{$\cost{\mathcal{T}_b}$ (bits)} & {$\cost{\mathcal{T}}$ (bits)} & {$\frac{\cost{\mathcal{T}}}{\cost{\mathcal{T}_b}} (\%)$} & 
{\em \# trees} & {\em \# sets} &&
{\em min--sup} & {\em \# sets} & {\em \# bits} & {\em ratio (\%)} \\
\midrule
{\em anneal} & 23104 & 12342 & {\bf 53.4} & {71} & {1203}       && {1} & {102} & {22154} & {34.6} \\
{\em breast} & 8099 & 2998 & {\bf 37.0} & {16} & {17}           && {1} & {30} & {4613} & {16.9} \\
{\em courses} & 76326 & 61685 & {\bf 80.8} & {98} & {1230}      && {2} & {148} & {71019} & {79.3} \\
{\em mammals} & 78044 & 50068 & {\bf 64.2} & {40} & {845}       && {200} & {254} & {90192} & {42.3} \\
{\em mushroom} & 442062 & 115347 & {\bf 26.1} & {119} & {999}   && {1} & {424} & {231877} & {20.9} \\
{\em nursery} & 337477 & 180803 & {\bf 53.6} & {32} & {3409}    && {1} & {260}  & {258898} & {45.5} \\
{\em pageblocks} & 15280 & 7611 & {\bf 49.8} & {44} & {219}     && {1} & {53}  & {10911} & {5.0} \\
{\em tic--tac--toe} & 25123 & 14137 & {\bf 56.3} & {29} & {619} && {1} & {162} & {28812} & {62.3} \\
\bottomrule
\end{tabular}
\vspace*{-0.5mm}
\end{center}
\end{table*}

If we apply \textsc{GreedyPack}, our greedy decision tree building method,
to these datasets we see that it is unable to compress
the independent database at all. Opposing, the dependently
generated dataset can be compressed into only 50\%
of the original number of bits. Inspection of 
the resulting itemsets show that the resulting model
correctly describes the dependencies in detail: 
The resulting $19$ itemsets are $\left\{a_1,\ldots,a_{10}, a_1a_2, \ldots, a_9a_{10}\right\}$.

\subsection{The Greedy Method}

Recall that our goal is to find high quality 
descriptions of the data. Following the MDL principle, 
the quality of the found descriptions can objectively
be measured by the compression of the data. We present the compressed sizes for \textsc{GreedyPack} in Table~\ref{table:greedy}.
The encoding costs $\cost{\mathcal{T}}$ include the size of the 
encoded data and the decision trees.
The initial costs, as denoted by $\cost{\mathcal{T}_b}$, are those of 
encoding the data using na\"{i}ve single-node {\sc TrivialTree}s.
Each of these experiments required 1--10 seconds runtime,
with an exception of $60$s for {\em mushroom}.

From Table~\ref{table:greedy}, we see that all models
returned by \textsc{GreedyPack} strongly reduce the 
number of bits required to describe the data;
this implicitly shows that good models are returned. 
The quality can be gauged by taking the compression ratios into
account. In general, our greedy method reduces the 
number of bits to only half of what the independent
model requires.
As two specific examples of the found dependencies, in the \emph{courses}
dataset the course \emph{Data Mining} was packed using \emph{Machine Learning},
\emph{Software Engineering}, \emph{Information Retrieval Methods} and \emph{Data Warehouses}.
Likewise, \emph{AI} and \emph{Machine Learning} were used
to pack the \emph{Robotics} course. 

Like discussed above, our approach and the {\sc Krimp}~\cite{siebes06}
algorithm have stark differences in what part of the data 
is considered. However, as both methods use compression, and 
result good itemsets, it is insightful to compare
the algorithms. For the latter we here allow it 
to compress as well as possible, and thus, consider
candidates up to as low min-sup thresholds as feasible.

Let us compare between the outcomes of either method.
For {\sc Krimp} these are itemsets, for ours it is the
combination of the decision trees and the related itemsets.
We see that {\sc Krimp} typically returns
fewer itemsets than \textsc{GreedyPack}. However, 
our method returns itemsets that describe interactions
between both present {\em and} absent items. 

Next, we observed that especially the initial {\sc Krimp}
compression requires many more bits than ours, and as such
{\sc Krimp} attains better compression ratios. However,
if we disregard the ratios and look at the raw number of
bits the two methods require, we see that {\sc Krimp} generally
requires twice as many bits to describe {\em only} the 1's
in the data than \textsc{GreedyPack} does to represent {\em all} 
of the data.

\subsection{Validation through Classification}
To further assess the quality of our models we use 
a simple classification scheme~\cite{leeuwen06}.
First, we split the training database into separate 
class-databases. We pack each of these. 
Next, the class labels
of the unseen transactions were assigned according to
the model that compressed it best.

We ran these experiments for three databases,
viz. {\em mushroom}, {\em breast} and {\em anneal}.
A random 90\% of the data was used to train the models, 
leaving 10\% to test the accuracy on.
The accuracy scores we noted, resp. 100\%,
98.0\% and 93.4\%, are fully comparable to 
(and for the second, even better than) 
the classifiers considered in~\cite{leeuwen06}.

\subsection{Choosing Good Itemsets}

\begin{table*}[htb!]
\begin{center}
\caption{Compressed sizes and number of extracted itemsets for the itemset selection algorithms.}\label{table:choose}
\begin{tabular}[b]{l rr@{\hspace{3mm}}c@{\hspace{3mm}}rrr@{\hspace{3mm}}c@{\hspace{3mm}}rrr@{\hspace{3mm}}c@{\hspace{3mm}}rr}
\toprule
& 
\multicolumn{2}{l}{\em Candidate Itemsets} & &
\multicolumn{3}{l}{\sc SetPack} & &
\multicolumn{3}{l}{\sc SetPackGreedy} & &
\multicolumn{2}{l}{\sc Krimp} \\
\cmidrule{2-3}
\cmidrule{5-7}
\cmidrule{9-11}
\cmidrule{13-14}
{\em Dataset} &
{\em min-sup} &{\em \# sets} & &
{$\cost{\mathcal{T}}$} & {$\frac{\cost{\mathcal{T}}}{\cost{\mathcal{T}_b}} (\%)$} & {\em \# sets} &&
{$\cost{\mathcal{T}}$} & {$\frac{\cost{\mathcal{T}}}{\cost{\mathcal{T}_b}} (\%)$} & {\em \# sets} &&
{\em \# bits} & {\em \# sets} \\
\midrule
{\em anneal} & 
	{175} & {8837} &&
	{20777} & {\bf 89.9} & {103} &&
	{20781} & {\bf 89.9} & {69} &&
	{31196} & {53} \\
{\em breast} & 
	{1} & {9920} &&
	{5175} & {\bf 63.7} & {42} &&
	{5172} & {\bf 63.9} & {49} &&
	{4613} & {30}\\
{\em courses} & 
	{55} & {5030} &&
	{64835} & {\bf 84.9} & {268} &&
	{64937} & {\bf 85.1} & {262} &&
	{73287} & {93}\\
{\em mammals} & 
	{700} & {7169} &&
	{65091} & {\bf 83.4} & {427} &&
	{65622} & {\bf 84.1} & {382} &&
	{124737} & {125} \\
{\em mushroom} & 
	{1000} & {123277}& &
	{313428} & {\bf 70.9} &{636} &&
	{262942} & {\bf 59.5} & {1225} &&
	{474240} & {140} \\
{\em nursery} & 
	{50} & {25777} &&
	{314081} & {\bf 93.0} & {276} &&
	{314295} & {\bf 93.1} & {218} &&
	{265064} & {225} \\
{\em pageblocks} & 
	{1} & {63599} &&
	{11961} & {\bf 78.3} & {92} &&
	{11967} & {\bf 78.3} & {95} &&
	{10911} & {53} \\
{\em tic--tac--toe} & 
	{7} & {34019} &&
	{23118} & {\bf 92.0} & {620} &&
	{23616} & {\bf 94.0} & {277} &&
	{28957} & {159} \\
\bottomrule
\end{tabular}
\end{center}
\end{table*}

In this subsection we evaluate {\sc SetPack}, our itemset selection algorithm. 
Recall that this algorithm selects itemsets such that they allow
for building succinct encoding decision trees. The difference with
\textsc{GreedyPack} is that in this setup the resulting itemsets
should be a subset of a given candidate family.
Here, we consider frequent itemsets as candidates.
We set the support threshold such that the experiments with {\sc SetPack}
were finished within $1 \over 2$--2 hours, with an exception of 23 hours for 
considering the large candidate family for {\em mushroom}.
For comparison
we use the same candidates for {\sc Krimp}. We also
compare to {\sc SetPackGreedy}, which required 1--12 minutes, 
7 minutes typically, with an exception of $2$$1 \over 2$ hours for {\em mushroom}.

Comparing the results of this experiment (Table~\ref{table:choose}) with the
results of \textsc{GreedyPack} in the previous experiment, we see that the
selection process is more strict: now even fewer itemsets are regarded as
interesting enough. Large candidate collections are strongly reduced in number:
up to three orders of magnitude. On the other hand, the compression ratios are
still very good. The reason that \textsc{GreedyPack} produces smaller
compression ratios is because it is allowed to consider any itemset.

Further, the fact alone that even with this very strict selection
the compression ratios are generally well below 90\% show that these few 
sets are indeed of high importance to describing the major interactions 
in the data. 

If we compare the number of selected sets to {\sc Krimp}, we see that our method
returns in the same order as many itemsets.
These descriptions require far less bits than those found
by {\sc Krimp}. As such, ours are a better approximation of the Kolmogorov 
complexity of the data.

Between \textsc{SetPack} and \textsc{SetPackGreedy} the outcomes
are very much alike; this goes for both the obtained compression 
as well as the number of returned itemsets. However, the greedy search
of \textsc{SetPackGreedy} allows for much shorter running times.

%% file: conclusions.tex
\section{Discussion}\label{section:discussion}

The experimentation on our methods validates the quality
of the returned models. The models correctly detect 
dependencies in the data while ignoring independencies.
Only a small number of itemsets is returned, which 
are shown to provide strong compression of the data. 
By the MDL principle we then know these describes all 
important regularities in the data distribution in
detail efficiently and without redundancy.
This claim is further supported by the high
classification accuracies our models achieve.

The {\sc GreedyPack} algorithm generally uses more itemsets and obtains
better packing ratios than {\sc SetPack}. While {\sc GreedyPack} is allowed to
use any itemset, {\sc SetPack} may only use frequent itemsets.
This suggests that we may able to achieve better ratios if we use
different candidates, for example, low-entropy sets~\cite{heikinheimo07entropy}.

The running times of the experiments reported in this work
range from seconds to hours and depend mainly on the
number of attributes and rows of the datasets. 
The exhaustive version \textsc{SetPack} may be slow on very large
candidate sets, however, the greedy version 
\textsc{SetPackGreedy} can even handle such families well.
Considering that our current implementation is 
rather na\"{i}ve and the fact that both methods are easily parallelized,
both \textsc{GreedyPack} and \textsc{SetPackGreedy}
are suited for the analysis of large databases.

The main outcomes of our models are the
itemsets that identify the encoding paths. However,
the decision trees from which these sets are extracted
can also be regarded as interesting as these 
provide an easily interpretable view on the major
interactions in the data. Further, just considering
the attributes used in such a tree as an itemset
also allows for simple inspection of the main associations.

In this work we employ the MDL criterion to identify the optimal 
model. Alternatively, one could consider using either BIC or 
AIC, both of which can easily be applied to judge between
our decision tree-based models. 

\section{Conclusions}\label{section:conclusion}

In this paper we presented two methods that find compact
sets of high quality itemsets. Both methods employ
compression to select the group of patterns that
describe {\em all} interactions in the data best.
That is, the data is considered symmetric and thus
both the 0s and 1s are taken into account in these 
descriptions.
Experimentation with our methods showed that high quality
models are returned. Their compact size, typically tens to
thousands of itemsets, allow for easy further analysis of 
the found interactions.